\documentclass[11pt]{article}
\usepackage{paralist}
\usepackage{color}
\usepackage{bm}

\usepackage{soul}

\usepackage[cm]{fullpage}
\usepackage{amsfonts}
\usepackage{amssymb}

\usepackage{amsmath}
\usepackage{constants}
\usepackage{amsthm}
\makeatletter
\newtheorem*{rep@theorem}{\rep@title}
\newcommand{\newreptheorem}[2]{%
	\newenvironment{rep#1}[1]{%
		\def\rep@title{#2 \ref{##1}}%
		\begin{rep@theorem}}%
		{\end{rep@theorem}}}
\makeatother

\usepackage{ifpdf}
\newcommand{\mydriver}{hypertex}
\ifpdf
\renewcommand{\mydriver}{pdftex}
\fi
\usepackage[breaklinks,\mydriver]{hyperref}

\usepackage[margin=1in]{geometry}

\theoremstyle{plain}
\newtheorem{theorem}{Theorem}[section]
\newtheorem{fact}[theorem]{Fact}
\newtheorem{lemma}[theorem]{Lemma}

\newtheorem{claim}[theorem]{Claim}
\newtheorem{definition}[theorem]{Definition}

\usepackage{algorithm}
\usepackage{algpseudocode}
\usepackage[normalem]{ulem}

\algtext*{EndWhile}
\algtext*{EndIf}
\algtext*{EndFor}

\DeclareMathOperator{\mincut}{mincut}
\DeclareMathOperator{\capacity}{cap}
\DeclareMathOperator{\lca}{lca}
\DeclareMathOperator{\con}{cong}
\DeclareMathOperator{\load}{load}

\newenvironment{pfof}[1]{\begin{proof}[\textbf{Proof of #1: }]}{\end{proof}}

\theoremstyle{definition}

\title{Vertex Sparsification in Trees\footnote{An extended abstract will appear in Proceedings of the 14th Workshop on Approximation and Online Algorithms (WAOA) 2016.}}
\author{Gramoz Goranci\footnote{University of Vienna, Faculty of Computer Science, Vienna, Austria. E-mail: \texttt{gramoz.goranci@univie.ac.at}. Part of this work was done while the author was a master's student at TU M\"unchen.}
	\and
	Harald R\"acke\footnote{Institut f\"ur Informatik, Technische Universit\"at M\"unchen, Garching, Germany. E-mail: \texttt{raecke@in.tum.de}.}
}
\date{\today}
\begin{document}
\begin{titlepage}
\maketitle

\thispagestyle{empty}
\begin{abstract}
Given an unweighted tree $T=(V,E)$ with terminals $K \subset V$, we show how to obtain a
$2$-quality vertex flow and cut sparsifier $H$ with $V_H = K$. We prove that our result is essentially tight by providing a~$2-o(1)$ lower-bound on the quality of any cut sparsifier for stars.

In addition we give improved results for quasi-bipartite graphs. First, we show how to obtain a $2$-quality flow sparsifier with $V_H = K$ for such graphs. We then consider the other extreme and construct exact sparsifiers of size $O(2^{k})$, when the input graph is unweighted.
\end{abstract}

\end{titlepage}
\section{Introduction}
Graph sparsification is a technique to deal with large input graphs by
``compressing'' them into smaller graphs while preserving important
characteristics, like cut values, graph spectrum etc. Its algorithmic value is apparent, since these smaller representations can be computed in a preprocessing step of an algorithm, thereby greatly improving performance. 

Cut sparsifiers (\cite{BenczurK96}) and spectral sparsifiers (\cite{SpielmanT04}) aim at
reducing the number of edges of the graph while approximately preserving cut
values and graph spectrum, respectively. These techniques are used in a variety
of fast approximation algorithms, and are instrumental in the development of
nearly linear time algorithms.

In vertex sparsification (\cite{HagerupKNR98,moitra09,leighton,raecke2014,distancepreserving,KammaKN2015,cheung2016}), apart from reducing the number of edges, the goal is also to reduce the number of vertices of a graph. In such setting,
one is given a large graph $G=(V,E,c)$, together with a relatively small subset
of terminals $K\subseteq V$. The goal is to shrink the graph while preserving
properties involving the terminals. For example, in \emph{Cut Sparsification} one wants to construct a graph $H=(V_H,E_H,c_H)$ (with
$K\subseteq V_H$) such that $H$ preserves mincuts between
terminals up to some approximation factor $q$ (the \emph{quality}). 

Hagerup et al.~\cite{HagerupKNR98} introduced this concept under the term
\emph{Mimicking Networks}, and focused on constructing a (small) graph $H$ that
maintains mincuts exactly. They showed that one can obtain $H$ with
$O(2^{2^{k}})$ vertices, where $k=|K|$. Krauthgamer et al.~\cite{KrauthgamerR13} and Khan et al.~\cite{KhanR14} independently proved that $2^{\Omega(k)}$ vertices are required for some graphs if we want to preserve mincuts exactly.

Moitra~\cite{moitra09} analyzed the setting where the graph $H$ is as
small as possible, namely $V_H=K$. Under this condition, he obtained a
quality $O(\log k/\log\log k)$ cut sparsifier. A lower bound of
$\Omega(\sqrt{\log k}/\log\log k)$ was presented by Makarychev et al.~\cite{mm10}. A strictly stronger notion
than a cut sparsifier, is a \emph{flow sparsifier} that aims at (approximately)
preserving all multicommodity flows between terminals. The upper bound
of~\cite{moitra09} also holds for this version, but the lower
bound is slightly stronger: $\Omega(\sqrt{\log k / \log\log k})$.

Due to the lower bounds on the quality of sparsifiers with $V_H=K$, the recent focus has been on obtaining better guarantees with slightly larger sparsifiers.
Chuzhoy~\cite{juliasteiner} obtained a constant quality flow sparsifier of size
$C^{O(\log\log C)}$, where $C$ is the total weight of the edges incident to
terminal nodes. Andoni et al.~\cite{andoni} obtained quality of $(1+\varepsilon)$ and
size $O(\operatorname{poly}(k/\varepsilon))$ for \emph{quasi-bipartite} graphs,
i.e., graphs where the terminals form an independent set. This is interesting
since these graphs serve as a lower bound example for Mimicking Networks, i.e.,
in order to obtain an exact sparsifier one needs size at least $2^{\Omega(k)}$.

In this paper we study flow and cut sparsifiers for trees. Since, for tree
networks it is immediate to obtain a sparsifier of size $O(k)$ and quality $1$,
we consider the problem of designing flow and cut sparsifiers with $V_H=K$ as
in the original definition of Moitra. In
Section~\ref{section: upper-bound} we show how to design such a flow sparsifier for unweighted
trees with quality $2$. In Section~\ref{section: lower-bound} we prove that this result is essentially tight by establishing a lower bound. Concretely, we prove that even for unweighted stars it is not possible to obtain cut sparsifiers with quality $2-o(1)$.

As a further applicaton of our techniques, we apply them to quasi-bipartite
graphs (Section \ref{section: applications}). We first obtain a 2-quality flow sparsifier with $V_H=K$ for such
graphs. In addition we explore the other extreme and construct exact sparsifiers
of size $O(2^k)$, if the input graph is unweighted. This shows that even though quasi-bipartite graphs serve as
lower bound instances for Mimicking Networks they are not able to close the
currently large gap between the upper bound of $O(2^{2^k})$ and the lower bound
of $2^{\Omega(k)}$ on the size of Mimicking Networks.

Finally we obtain hardness results for the problem of deciding whether a
graph $H$ is a sparsifier for a given unweighted tree $T$. We prove that this
problem is co-$\mathcal{NP}$-hard for cut sparsifiers, based on Chekuri et al.~\cite{ChekuriSOS07}. For flow sparsifiers we show that for
a single-source version, where the sparsifier has to preserve flows in which all
demands share a common source, the problem is co-$\mathcal{NP}$-hard. See Section \ref{section: hardness} for more details.

\section{Preliminary}

Let $G=(V,E,c)$ be an undirected graph with terminal set $K \subset V$ of
cardinality $k$, where $c: E \rightarrow \mathbb{R}^{+}$ assigns a non-negative
capacity to each edge.
We present two different ways to sparsify the number of vertices in
$G$. 

Let $U \subset V$ and $S \subset K$. We say that a cut $(U, V \setminus U)$ is
$S$-separating if it separates the terminal subset $S$ from its complement $K
\setminus S$, i.e., $U \cap K$ is either $S$ or $K \setminus S$. The cutset
$\delta(U)$ of a cut $(U, V \setminus U)$ represents the edges that have one
endpoint in $U$ and the other one in $V \setminus U$. The cost
$\capacity_G(\delta(U))$ of a cut $(U, V \setminus U)$ is the sum over all
capacities of the edges belonging to the cutset. We let $\text{mincut}_{G}(S, K
\setminus S)$ denote the $S$-separating cut of the minimum cost in
$G$. A graph $H = (V_H, E_H, c_H)$, $K \subset V_H$ is a \emph{vertex cut
  sparsifier} of $G$ with \emph{quality} $q \geq 1$ if: $ \forall S \subset K,
~ \mincut_G(S, K \setminus S) \leq \mincut_H(S, K \setminus S) \leq q \cdot
\mincut_G(S, K \setminus S). $

We say that a (\textit{multi-commodity}) flow $f$ is a \textit{routing} of the
demand function $d$, if for every terminal pair $(x,x')$ it sends $d(x,x')$ units of
flow from $x$ to $x'$. The \textit{congestion} of an edge $e \in E$ incurred by
the flow $f$ is defined as the ratio of the total flow sent along the edge to
the capacity of that edge, i.e., $f(e)/c(e)$. The \textit{congestion of the
  flow} $f$ for routing demand $d$ is the maximum congestion over all edges in
$G$. We let $\con_{G}(d)$ denote the minimum congestion over all flows. 
A graph $H = (V_H, E_H, c_H)$, $K \subset V_H$ is a \emph{vertex flow
  sparsifier} of $G$ with \emph{quality} $q \geq 1$ if for every demand function
$d$, $\con_H(d) \leq \con_G(d) \leq q \cdot \con_H(d). $

We use the following tools about sparsifiers throughout the paper.
\begin{lemma}[\cite{leighton}] \label{lemma: Moitra} If $H = (V_H, E_H, c_H)$, $V_H=K$ is a
  vertex flow sparsifier of $G$, then the quality of $H$ is $q =
  \con_{G}(d_{H})$, where $d_{H}(x,x') : = c_{H}(x,x')$ for all terminal pairs
  $(x,x')$.
\end{lemma}
Let $G_1$ and $G_2$ be graphs on \textit{disjoint} set of vertices with
terminals $K_1 = \{s_1,\ldots,s_k\}$ and $K_2 = \{t_1,\ldots,t_m\}$,
respectively. In addition, let $\phi(s_i) = t_i$, for all $i =1,\ldots,\ell$,
be a one-to-one correspondence between some subset of $K_1$ and $K_2$. The
$\phi$-\emph{merge} (or $2$-\textit{sum}) of $G_1$ and $G_2$ is the graph $G$
with terminal set $K = K_1 \cup \{t_{\ell+1},\ldots,t_m\}$ formed by
identifying the terminals $s_i$ and $t_i$ for all $i = 1,\ldots, \ell$. This
operation is denoted by $G := G_1 \oplus_{\phi} G_2$.
\begin{lemma}[\cite{andoni}, Merging] \label{lemma: Merging} Let $G = G_1 \oplus_{\phi} G_2$.
  Suppose $G_1'$ and $G_2'$ are flow sparsifiers of quality $q_1$ and $q_2$ for
  $G_1$ and $G_2$, respectively. Then $G' = G_1' \oplus_{\phi} G_2'$ is a flow
  sparsifier of quality $\max\{q_1, q_2\}$ for $G$.
\end{lemma} 
\begin{lemma}[Convex Combination of Sparsifiers] \label{lemma: ConvexComb} Let
  $H_i = (V^{*},E_i, c_i)$, $i = 1, \ldots, m$ with $K \subset V^{*}$ be vertex
  flow sparsifiers of $G$. In addition, let $\alpha_1, \alpha_2, ..., \alpha_m$
  be convex multipliers corresponding to $H_i$'s such that $\sum_{i} \alpha_i =
  1$. Then the graph $H' = \sum_{i} \alpha_i \cdot H_i $ is a vertex flow
  sparsifier for $G$.
\end{lemma}
 
\section{Improved Vertex Flow Sparsifiers for Trees} \label{section: upper-bound}
In this section we show that given an unweighted tree $T=(V,E)$, $K \subset V$, we can
construct a flow sparsifier $H$ only on the terminals, i.e., $V(H)=K$, with
quality at most $4$. We then further improve the quality to $2$. The graph $H$ has the nice property of being a convex combination of trees. 

We obtain the quality of $4$ by combining the notion of probabilistic mappings due to Andersen and Feige~\cite{AF09} and a duality argument due to R\"acke~\cite{raecke08}. Our result then immediately follows using as a black-box an implicit result of Gupta~\cite{gupta01}.
We note that a direct application of the Transfer Theorem due to Andersen and
Feige~\cite{AF09} does not apply, since their interchangeability argument relies on
arbitrary capacities and lengths.

Let $w : E \rightarrow \mathbb{R}_{\geq 0}$ be a function which assigns
non-negative values to edges which we refer to as \emph{lengths}. Given a tree
$T=(V,E,w)$ we use $d_{w} : V \times V \rightarrow \mathbb{R}_{\geq 0}$ to
denote the shortest path distance induced by the edge length $w$. A
$0$-\emph{extension} of a tree $T = (V,E)$, $K \subset V$ is a retraction $f :
V \rightarrow K$ with $f(x) = x$, for all $x \in K$, along with another graph
$H = (K, E_{H})$ such that $E_{H} = \{(f(u),f(v)) : (u,v) \in E\}$. The graph
$H$ is referred to as a \emph{connected} $0$-\emph{extension} if in addition we
require that $f^{-1}(x)$ induces a connected component in $T$.

Given a graph $G = (V,E)$, we let $\mathcal{P}$ be a collection of multisets of
$E$, which will be usually referred to as paths. A mapping $M : E \rightarrow
\mathcal{P}$ maps every edge $e$ to a path $P \in \mathcal{P}$. This mapping
can be alternatively represented as a non-negative square matrix $M$ of
dimension $|E| \times |E|$, where $M(e',e)$ is the number of times edge $e$
lies on the path $M(e')$. Let $\mathcal{M}$ denote the collection of mappings
$M$. If we associate to each mapping $M \in \mathcal{M}$ a convex multiplier
$\lambda_M$, the resulting mapping is referred to as a \textit{probabilistic
  mapping}.

\medskip

\noindent \textbf{Connected $0$-extension embedding on Trees. } Suppose we are
given a tree $T = (V,E)$, $K \subset V$ and a connected $0$-extension $(H, f)$,
where $H = (K,E_{H})$ and $f$ is a retraction. Given an edge $(u,v) \in E$ from
$T$, we can use the retraction $f$ to find the edge $(f(u), f(v))$ in $H$ (if
$u$ and $v$ belong to different components). Since this edge is not an edge of
the original tree $T$, we need a way to map it back to $T$ in order to be
consistent with our definition of mappings. The natural thing to do is to take
the unique shortest path between $f(u)$ and $f(v)$ in $T$. Denote by $S^{T}_{u,v}$
all the edges in the shortest path between $u$ and $v$ in $T$. Then, we let
$M_{H,f}((u,v)) = S^{T}_{f(u),f(v)}$ be the mapping $M_{H,f} : E \rightarrow
\mathcal{P}$ induced by $(H,f)$.

Let $\mathcal{H}$ be the family of all connected $0$-extensions for $T$, which
are also trees. We then define the collection of mappings $\mathcal{M}$ for $T$
by $\{M_{H,f} : H \in \mathcal{H}\}$.

\medskip

\noindent
\textbf{Capacity Mappings. }Given a tree $T = (V,E,c)$, $c : E \rightarrow
\mathbb{R}^{+}$ and a connected $0$-extension $(H,f)$, the \textit{load} of an
edge $e \in E$ under $(H,f)$ is $\load_f(e) = \sum_{e'} M_{H,f}({e',e}) \cdot
c(e')$. The \textit{expected load} of an edge $e \in E$ under a probabilistic
mapping is $\sum_{i}\lambda_i \load_{f_{i}}(e)$.

\medskip

\noindent
\textbf{Distance Mappings. }Given a tree $T = (V,E,w)$, $w : E \rightarrow
\mathbb{R}^{+}$ and a connected $0$-extension $(H,f)$, the \emph{mapped length}
of an edge $e' = (u',v') \in E$ under $(H,f)$ is $d_w(f(u'), f(v')) = \sum_{e}
M_{H,f}({e',e}) \cdot w(e)$. The \textit{expected mapped length} of an edge
$e'=(u',v') \in E$ under a probabilistic mapping is $\sum_{i}\lambda_i
d_w(f_i(u'), f_i(v'))$. \medskip

With the above definitions in mind, for some given tree $T=(V,E,c)$, we can
find a flow sparsifiers that is a convex combination of connected
$0$-extensions using the following linear program, and its dual.

\noindent
\begin{minipage}{7.3cm}%
\begin{equation*}\label{primal}%
\begin{array}{@{}ll@{}rl}
 \text{min}  &  \alpha &  &   \\[0.2cm]
 \text{s.t.}  & \forall e \quad &   \sum_{i} \nolimits \lambda_i \cdot \load_{f_i}(e) &\leq \alpha \cdot c(e)   \\[0.3cm]
  & &  \sum_{i} \nolimits \lambda_i & \geq 1  \\[0.3cm]
  & \forall i &  \lambda_i & \geq 0.   
\end{array}
\end{equation*}
\end{minipage}
\hfill
\begin{minipage}{7.3cm}%
\begin{equation*}\label{dual}%
\begin{array}{@{}ll@{}rl}
 \text{min}  & \displaystyle \beta &  &   \\[0.2cm]
 \text{s.t.}  & \forall i \quad &   \sum_{e} \nolimits w(e) \cdot \load_{f_i}(e) &\geq \beta  \quad (*) \\[0.3cm]
  & &  \sum_{e} \nolimits w(e) \cdot c(e)& \leq 1  \\[0.3cm]
  & \forall e &  w(e) & \geq 0.  
\end{array}
\end{equation*}
\end{minipage}
\medskip

\noindent Next, we re-write the dual constraints of type $(*)$ as follows:
\begin{equation*} 
\begin{split}
 \sum_{e} \nolimits w(e)  \load_{f_i} (e)  & =    \sum_{e} \nolimits w(e)  \sum_{e'} \nolimits M_{H,f_i}({e',e}) \cdot c(e')  \\[1ex]
   & =   \sum_{e'} \nolimits c(e') \left( \sum_{e} \nolimits M_{H,f_i}(e',e) \cdot w(e) \right) =  \sum_{e' = (u',v')} \nolimits c(e') \cdot d_{w}(f_i(u'),f_i(v')) \enspace.
\end{split}
\end{equation*}
\noindent
Using this re-formulation and a few observations, the dual is equivalent to:
\begin{equation} \label{newDual}
	\max_{w \geq 0} \min_{i} { \sum_{e=(u,v)} \nolimits c(e) \cdot d_w(f_i(u),f_i(v))}~/~{ \sum_{e} \nolimits w(e) \cdot c(e)} \enspace.
\end{equation}
For the unweighted case $c(e)=1$, we can make use of the following lemma:  
\begin{lemma}[{\cite[Lemma 5.1]{gupta01}}] \label{lemmaGupta}
Given a tree $T=(V,E,w)$, $K \subset V$, we can find a connected $0$-extension $f$ such that $\sum_{e=(u,v)} \nolimits d_w(f(u),f(v)) \leq 4 \cdot  \sum_{e} w_e.$
\end{lemma}
The above lemma tells us that optimal value of (\ref{newDual}) is bounded by
$4$. This implies that the optimal value of the dual is bounded by $4$, and by
strong duality, the optimal value of the primal is also bounded by $4$. The
latter implies that $T$ admits a $4$-quality vertex sparsifier of size $k$.

\mathversion{bold}
\subsection{Obtaining Quality $2$}
\mathversion{normal}
Next we show how to bring down the quality of flow sparsifiers on
trees to $2$. We give a direct algorithm that constructs a flow sparsifiers and
unlike in the previous subsection, it does not rely on the interchangeability
between distances and capacities. We first consider trees where terminals are
the only leaf nodes, i.e., $L(T)=K$. Later we show how to extend the
result to arbitrary trees.

To convey some intuition, we start by presenting the deterministic version of
our algorithm. We maintain at any point of time a partial mapping $f$--setting
$f(v) = \perp$, when $f(v)$ is still undefined, but producing a valid connected
$0$-extension when the algorithm terminates. Note that $f(x)=x$, for all $x \in
K$. Without loss of generality, we may assume that the tree is rooted at some
non-terminal vertex and the child-parent relationships are defined. The
algorithm works as follows: it repeatedly picks a non-terminal $v$ farthest
from the root and maps it to one of its children $c$, i.e.,
$f(v)=f(c)$\footnote{Alternatively, one can view this step as contracting an
  arbitrary child-edge of $v$.} (we refer to such procedure as Algorithm 1). This process results in a flow sparsifier that is a connected
$0$-extension.

Unfortunately, the quality of the sparsifier produced by the above algorithm
can be very poor. To see this, consider an unweighted star graph $S_{1,k}$,
where leaves are the terminal vertices and the center is the non-terminal
vertex $v$. Any connected $0$-extension of $S_{1,k}$ is a new star graph
$S_{1,k-1}$ lying on the terminals, where the center is the terminal $x$ with
$f(v) = x$. Now, consider a demand function $d$ that sends a unit flow among all
edges in $S_{1,k-1}$. Clearly, $d$ can be feasibly routed in $S_{1,k-1}$. But
routing $d$ in $S_{1,k}$ gives a load of at least $k-1$ along the edge $(x,v)$,
and thus the quality of $S_{1,k-1}$ is at least $k-1$ (Lemma \ref{lemma:
  Moitra}).

One way to improve upon the quality is to map the non-terminal $v$ uniformly at
random to one of the terminals. We can equivalently view this as taking convex
combination over all possible connected $0$-extensions of
$S_{1,k}$. 
By Lemma \ref{lemma: ConvexComb} we know
that such a convex combination gives us another flow sparsifier for $S_{1,k}$,
and it can be checked that the quality of such a sparsifier improves to $2$.
Surprisingly, we show that applying this trivial random-mapping of
non-terminals in trees with terminals as leaves leads to a flow sparsifier $H$ which is a random connected $0$-extension and achieves similar guarantees. This procedure is summarized in Algorithm \ref{algo: contraction2}.

\addtocounter{algorithm}{1}
\algnewcommand\algorithmicinput{\textbf{Input:   }}
\algnewcommand\Input{\item[\algorithmicinput]}
\begin{algorithm}
\caption{\textsc{Randomized Connected $0$-extension}}
\begin{algorithmic}[1]
\Input Tree $T=(V,E)$, $K$, $L(T) = K$.
\State Set $f(x)=x$ for all $x \in K$, $f(v)= \perp$ for all $v \in V \setminus K$.
\While{there exists a $v$ such that $f(v)=\perp$}
\State Choose a non-terminal $v$ farthest from the root and let $C(v)$ be its children.
\State Set $f(v) = f(c)$, where $c \in C(v)$ is chosen uniformly at random. \EndWhile
\end{algorithmic}
\label{algo: contraction2}
\end{algorithm}
\begin{claim} \label{claim: connected0extRand}For a tree $T=(V,E)$, $K \subset
  V$, $L(T) = K$, Algorithm \ref{algo: contraction2} produces a flow sparsifier
  of $T$ that is a random connected $0$-extension $(H,f)$ with $H=(K,E_H)$.
  Moreover, $H = \sum_{i} \lambda_i \cdot H_i$, $\sum_{i} \lambda_i = 1$, where
  the sum is over connected $0$-extensions $(H_i,f_i)$ produced by Algorithm 1.
\end{claim}

To compute the quality of $H$ as a flow sparsifier for $T$, we need to bound
the congestion of every edge of $T$ incurred by the embedding of $H$ into $T$.
This embedding routes the capacity of every terminal edge $(x,x')$ in $H$ along
the (unique) shortest paths between leaves $x$ and $x'$ in $T$. First, we
crucially observe that without loss of generality, it suffices to bound the
load of the edges incident to the terminals, i.e., edges incident to leaf
vertices. To see this, let $(u,v)$ be an edge among non-terminals in $T$, with
$v$ being the parent of $u$. Now, when embedding $H$ into $T$, we know that the
demands among all terminal pairs that lie in the subtree $T(u)$ rooted at $u$
\emph{cannot} incur any load on the edge $(u,v)$, as these terminal shortest
paths do not use this edge. Thus, we can safely replace the subtree $T(u)$ with
some dummy terminal and perform the analysis as before.

First, we study edge loads under deterministic connected $0$-extensions. Let $e
= (x,v)$ be the edge incident to $x \in K$, $m_x$ denote the level of $x$ in
$T$ and $\{x,v_{m_x-1},\ldots,v_0\}$ be the set of vertices belonging to the
shortest path between $x$ and the root $r = v_0$ in $T$. Given a connected
$0$-extension $f_i$ output by Algorithm~1 
, we say that
$x$ is \emph{expanded} up to the $\ell$-th level if $f_i(v_j) = x$, for all $j
\in \{m_x,\ldots, \ell\}$. This leads to the following lemma.
\begin{lemma} \label{lemma: load} Let $e = (x,v)$ be the edge incident to $x
  \in K$, $(H_i,f_i)$ be a connected $0$-extension and recall that empty sum is
  defined as $0$. If $x$ is expanded up to the $\ell$-th level, then the load
  of $e$ under $(H_i,f_i)$ is
\[\load_{f_i}(e) \leq
    1 + \sum_{j=\ell}^{m_x-1}(c_j-1) ,~ \ell \in \{m_x,\ldots, 0\},\]
where $c_j$ denotes the number of children of  non-terminal $v_j$ in $T$. 
\end{lemma}   
Let $I^{x}_{\ell} = \{(H_i,f_i)\}$ be the set of connected $0$-extensions
output by Algorithm~1 
 where $x$ is expanded up to the
$\ell$-th level. We observe that the edge $e$ has the same load regardless of
which element of $I^{x}_{\ell}$ we choose. Thus, for any $(H_i,f_i) \in
I^{x}_{\ell}$, we can write $\load_\ell(e) =
\load_{f_i}(e)$. 

Now, we study the expected edge loads under the random connected $0$-extension
output by Algorithm~2. 
 Let $N$ be the number of all
different connected $0$-extensions that can be output by Algorithm~1. 
If by $Z^{x}_\ell$ we denote the event that $x$ is expanded up
to the $\ell$-th level, then 
it follows that the expected load $\mathbb{E}[\load_f(e)]$ of $e=(x,v)$ under
$(H,f)$ is
\begin{equation} \label{eq: exload}
	 \sum_{i=1}^{N} \load_{f_i}(e)/N  =  \sum_{\ell=0}^{m_x} \frac{\text{\# of }f_i\text{'s s.t. } Z^{x}_{\ell} }{N} \cdot \load_{\ell}(e) =  \sum_{\ell=0}^{m_x} \mathbb{P} [Z^{x}_{\ell}] \cdot \load_{\ell}(e)\enspace.
\end{equation}
Since in Algorithm~2 
 all non-terminals are mapped
independently of each other, we obtain $\mathbb{P}[Z^{x}_{\ell}] =
(1-1/c_{\ell-1}) \prod_{j=\ell}^{m_x-1} {1}/{c_{j}}, \ell \in
\{m_x,\ldots,1\}$ (recall that the empty product is defined as $1$). Further, observe that $\mathbb{P}[Z^{x}_0] = 1 / \prod_{j=0}^{m_x-1} c_j \enspace.$
Plugging the probabilities and Lemma \ref{lemma: load} in (\ref{eq: exload}), we get that 
$\mathbb{E}[\load_f(e)]$ is
\begin{equation*} 
\frac{1}{\prod_{j=0}^{m_x-1} c_j} \bigg(1 + \sum_{j=0}^{m_x-1} (c_j - 1)\bigg)  +  \sum_{\ell = 1}^{m_x} (1-1/c_{\ell-1})  \prod_{j=\ell}^{m_x-1} \frac{1}{c_{j}} \bigg(1 + \sum_{j = \ell}^{m_x-1} (c_j-1)\bigg)\enspace.
\end{equation*}
Next, we rewrite the above as $A / B$, where $B = \prod_{j=0}^{m_x-1}c_j$ and $A$ is given by
\begin{equation*} \label{eqn3.5} 1 + \sum_{j=0}^{m_x-1} (c_j - 1) +
  \sum_{\ell=1}^{m_x-1} (c_{\ell-1} - 1) \prod_{j=0}^{\ell-2} c_j \bigg(1 +
    \sum_{j=\ell}^{m_x-1}(c_j - 1) \bigg) + (c_{m_x-1} - 1)
  \prod_{j=0}^{m_x-2}c_j \enspace.
\end{equation*}
The following lemma simplifies the middle expression of $A$.
\begin{lemma} \label{lemma: Induction} For any positive integers $\{c_0,\ldots, c_{m_x-1}\}$ and  $m_x \geq 3$,
\begin{equation*} 
\sum_{\ell=1}^{m_x-1} (c_{\ell-1} - 1) \prod_{j=0}^{\ell-2}
    c_j \big(1 + \sum_{j=\ell}^{m_x-1}(c_j - 1) \big) = (c_{m_x-1} + 1)
    \prod_{\ell=0}^{m_x-2} c_\ell - \sum_{\ell = 0}^{m_x-1} (c_\ell - 1) - 2 \enspace.
\end{equation*}
\end{lemma}

\begin{proof}
  Let $P(m_x-1)$ be the left-hand side expression in the statement of the
  lemma. We proceed by induction on $m_x$. For the base case $m_x=3$, it is
  easy to argue that the claim is valid. If we assume that the lemma holds true
  for $m_x-1$, then we get that:
\begin{equation} \label{eqn 3.6}
\begin{aligned}
	P(m_x) = &  \sum_{\ell=1}^{m_x-1} (c_{\ell-1} - 1) \prod_{j=0}^{\ell-2} c_j \bigg(1 + \sum_{j=\ell}^{m_x-1}(c_j - 1) + (c_{m_x} - 1) \bigg) \\
	 & + (c_{m_x-1} -1) \prod_{j=0}^{m_x-2}c_j\bigg((c_{m_x}-1) + 1\bigg) \\
	 = & \sum_{\ell=1}^{m_x-1} (c_{\ell-1} - 1) \prod_{j=0}^{\ell-2} c_j \bigg(1 + \sum_{j=\ell}^{m_x-1}(c_j - 1) \bigg)  \\
	 & + (c_{m_x} - 1) \sum_{\ell=1}^{m_x} (c_{\ell-1} -1) \prod_{j=0}^{\ell-2}c_j \; + \;   (c_{m_x-1} -1) \prod_{j=0}^{m_x-2}c_j\enspace.
\end{aligned} 
\end{equation}
Note that the following expression is a simple telescoping series:
\begin{equation} \label{eqn3.7}
	\sum_{\ell = 1}^{m_x} (c_{\ell-1} -1) \prod_{j=0}^{\ell-2} c_j = \prod_{\ell = 0}^{m_x-1} c_\ell - 1\enspace.
\end{equation}
Plugging this into Eqn. \eqref{eqn 3.6} and using induction hypothesis gives:
\begin{equation*}
\begin{split}
 P(m_x) 
&=  (c_{m_x-1} + 1) \prod_{\ell=0}^{m_x-2} c_\ell - \sum_{\ell = 0}^{m_x-1} (c_\ell - 1) - 2 
\quad+  (c_{m_x} -1) \bigg( \prod_{\ell=0}^{m_x-1} c_\ell -1\bigg) \\
&\quad+ (c_{m_x-1} -1) \prod_{j=0}^{m_x-2} c_j  
= (c_{m_x} + 1) \prod_{\ell=0}^{{m_x}-1} c_\ell - \sum_{\ell = 0}^{m_x} (c_\ell - 1) - 2\enspace.
\end{split}
\end{equation*}
This completes the induction step, and hence the proof of the lemma.
\end{proof}

Now, plugging the above lemma in $A$ we get that $A = 2B-1$. Thus,
$\mathbb{E}[\load_f(e)] = (2B-1)/B \leq 2$. Since we consider only unweighted
trees, it follows that the expected congestion of every edge is also bounded by $2$.
Taking the maximum over all edge congestions yields the following:
\begin{lemma} \label{lemma: TreeTerminalLeaves} Given a tree $T=(V,E)$, $K \subset V$, $L(T)=K$, there is a $2$-quality flow sparsifier $H$, which is a convex combination over connected $0$-extensions of $T$.
\end{lemma}
\noindent
\textbf{Derandomization.} Next we show that Algorithm 2 
 can be easily derandomized. We obtain a deterministic
algorithm that runs $O(n+k^{2}\alpha(2k))$ time and gives the same
guarantees as in Lemma \ref{lemma: TreeTerminalLeaves}, where $\alpha(\cdot)$ is
the inverse Ackermann function.

We first give an $O(n)$ time preprocessing step. For a tree $T=(V,E)$, $K \subset V$,
$L(T)=K$, we repeatedly contract edges incident to non-terminals of degree $2$
in $T$. When all such non-terminals are deleted from $T$,
our new tree can have at most $2k$ vertices. Note that this tree exactly preserves all flows among terminals. 

Now, we crucially observe that in the flow sparsifier $H$ output by Algorithm~2,
the capacity between any two terminals $x$ and $x'$
is exactly the probability that $x$ and $x'$ are connected under the random
mapping $f$. 
We next show that this probability can be computed efficiently.

Let $(x,x')$ be any terminal pair, $\lca(x,x')$ denote their lowest common
ancestor in $T$ and $r$ denote the level of $\lca(x,x')$ in $T$. Moreover, let
$V^{x}_{r} = \{x,v_{m_x-1},\ldots,v_{r}\}$, $v_{r} = \lca(x,x')$, be the set of
vertices belonging to the shortest path between $x$ and the $\lca(x,x')$.
Similarly, define $V^{x'}_{r}=\{x',v'_{m_{x'}-1},\ldots,v_{r}\}$. Since in
Algorithm~2 
 all non-terminals are mapped independently
of each other, we obtain
\begin{equation} \label{eq: derandomized}
\begin{split}
	\mathbb{P}[(f(x),f(x')) \in E_H]  & =   2 \cdot \mathbb{P}[f(v_r) = x] \cdot \mathbb{P}[f(v)=x,~\forall v \in V^{x}_{r-1}]  \cdot \mathbb{P}[f(v')=x',~\forall v' \in V^{x'}_{r-1}] \\[1ex]
	 & = \frac{2}{c_r} \cdot \prod_{j=r}^{m_{x}-1} \frac{1}{c_j}  \prod_{j=r}^{m_{x'}-1} \frac{1}{c'_j} \enspace.
\end{split}
\end{equation}
where $c_j$, $c'_j$ are the number of children of the non-terminal $v_j$, $v'_j$, respectively. 

The above expression suggest that one should build an efficient data-structure
for $T$ that answers queries of the form ``What is the product of the elements
associated with vertices along the path from $x$ to $x'$ in $T$?''. This
problem is known as \emph{The Tree Product Query} problem. For an arbitrary
tree with $n$ vertices, Alon and Schieber \cite{AlonS87} show that in order to answer each
Tree Product query in at most $O(\alpha(n))$ steps, an $O(n)$ preprocessing
time is sufficient.

Now we are ready to give our deterministic procedure. We first apply our
initial preprocessing step in $O(n)$ time. Since the resulting tree has at most
$2k$ vertices, it takes $O(k)$ time to preprocess the tree such that every
internal vertex knows the number of its children. Next, using $O(k)$
preprocessing, we build a data-structure for the Tree Product Query problem.
Now, for every terminal pair $(x,x')$ we can compute in $O(\alpha(2k))$ time
the capacity of $(x,x')$ in $H$ from the Tree product query between $x$ and
$x'$ and Eqn. (\ref{eq: derandomized}). Since there are at most $O(k^{2})$
terminal pairs, we get a running time of $O(n + k^{2}\alpha(2k))$. The
correctness is immediate from the above observations.

\medskip

\noindent
\textbf{Extension to Arbitrary Trees.} We show that one can reduce vertex sparsificiation for arbitrary trees to trees having terminals as leaf nodes. First, observe that without loss of generality, $L(T) \subseteq K$. Indeed, if there is a non-terminal leaf vertex $u$, we can simply remove it as $u$ cannot belong to any shortest path between two terminals. Note that the resulting tree exactly preserves all multicommodity flows among terminals. We repeatedly remove such vertices until $L(T) \subseteq K$.

Now assume that $u \in K \setminus L(T)$, i.e., $u$ is an internal terminal vertex, let $\delta(u)$ denote its degree and let $v_1,\ldots,v_{\delta(u)}$ be its neighbours. We make $\delta(u)$ copies $u_1,\ldots, u_{\delta(u)}$ of vertex $u$ and replace a neighbour edge $(u,v_i)$ by $(u_i,v_i)$. This splits the tree $T$ into $\delta(u)$ trees $T_i$, $i=1,\ldots, \delta(u)$, each having a copy of $u$. We let $K_i = V(T_i) \cap K$ be the new terminal set for $T_i$. We recursively apply this procedure to each $T_i$ until the only leaf nodes in the resulting trees are terminals. We then invoke Lemma \ref{lemma: TreeTerminalLeaves} to each such tree and finally combine these sparsifiers by merging the copies of the terminal at which they previously split. An inductive argument along with Lemma \ref{lemma: Merging} shows that the quality of the final sparsifier can be again bounded by $2$. This leads to the following theorem:

\begin{theorem} Given an unweighted tree $T=(V,E)$, $K \subset V$, there exists a $2$-quality flow sparsifier $H$. Moreover, $H$ can be viewed as a convex combination over connected $0$-extensions of $T$.
\end{theorem}

\section{Lower Bound} \label{section: lower-bound}
In this section we present a $2 - o(1)$ lower bound on the quality of any cut sparsifier for a star graph. Since previous lower bounds relied on non-planar graph instances, this is the first non-trivial lower bound for arbitrary cut sparsifiers on planar graphs. The result extends to the stronger notion of flow sparsifiers.

The main idea behind our approach is to exploit the symmetries of the star graph. We observe that these symmetries induce other symmetries on the cut structure of the graph. This simplifies the structure of an optimal cut-sparsifier.

Let $G = (K \cup \{v\}, E)$, be an unweighted star with $k$ terminals. Let $\pi'$ be any permutation of $K$. We extend $\pi'$ to a permutation $\pi$ of $K \cup \{v\}$ by setting $\pi(x) = \pi'(x), \forall x \in K$ and $\pi(v) = v$. Now, for any $U \subset K \cup \{v\}$ and any such a permutation $\pi$, we use the symmetry $\capacity_G(\delta(U)) = \capacity_G(\delta(\pi(U)))$. The latter implies that for any $S \subset K$, $\mincut_G(S, K \setminus S) = \mincut_G(\pi(S), K \setminus \pi(S))$. 

For a cut sparsifier $H$ of quality $q$ for $G$, we show that $\pi(H)$, i.e., the graph obtained by renaming all vertices of $H$ according to permutation $\pi$, is also a cut sparsifier of quality $q$ for $G$. Indeed, for any $S \in K$, $\capacity_{\pi(H)}(\delta(S)) = \capacity_{H}(\delta(\pi^{-1}(S))) \geq \mincut_G(\pi^{-1}(S), K \setminus \pi^{-1}(S) ) = \mincut_G(S, K \setminus S)$. Symmetrically, one can show that $\capacity_{\pi(H)}(\delta(S)) \leq q \cdot \mincut_G(S, K \setminus S)$.

\begin{lemma} \label{lemma: ConvexCombCut} A convex combination of any two cut sparsifiers with the same quality gives a new cut sparsifier with the same or better quality. 
\end{lemma}

\begin{lemma}
For the star graph $G$ defined as above, there exists an optimum cut sparsifier $H$, which is a complete graph with uniform edges-weights.
\end{lemma}
\begin{proof}
First, we observe by Lemma \ref{lemma: ConvexCombCut} that if we have two cut sparsifiers with the same quality, taking their convex combination gives a new cut sparsifier with the same or better quality. Suppose we are given some optimum cut sparsifier $H'$. We can generate $k!$ different cut sparsifiers by considering all possible permutations $\pi$ as defined above. By the above arguments, for each $\pi$, we know that $\pi(H')$ is also an optimum cut sparsifier. Taking the convex combination over $k!$ such sparsifiers, we obtain a complete graph $H$ with uniform edge-weights. 
\end{proof}

\begin{lemma}
If $H$ is uniform weighted complete graph that is an optimum cut sparsifier for the star graph $G$ and $k$ even, the edge weight must be at least $2/k$. 
\end{lemma}
\begin{proof}
By definition, $H$ must dominate the terminal cut that has $k/2$ vertices on one side. The minimum value of such a cut in $G$ is $k/2$. The number of edges that cross such a cut in $H$ is $k^{2}/4$. Since $H$ has uniform edge-weights, this gives that the edge weight must be at least $2/k$.
\end{proof}

\begin{theorem} Let $G = (K \cup \{v\}, E)$ be an unweighted star with $k$ terminals. Then, there is no cut sparsifier $H$ that achieves quality better than $2-o(1)$.
\end{theorem}
\begin{proof}
By the above lemmas, we can assume without loss of generality that $H$ is a complete graph with uniform edge-weights, where this edge weight is at least $2/k$. Hence, a cut that has a singleton terminal vertex on one side has capacity $2(k-1)/k = 2(1-1/k)$ in $H$ but it has minimum cut value $1$ in $G$. The latter implies that the quality of $H$ must be at least $2(1-1/k)$.
\end{proof}

\section{Hardness of Vertex Sparsification in Trees} \label{section: hardness}

\subsection{Hardness of Cut Sparsifiers}
In this section we show that for a given graph $G=(V,E,c)$ with $K \subset V$, and $H=(V_H,E_H,c_H)$ with $V_H = K$, deciding whether $H$ is a cut sparsifier of $G$ is co-$\cal{NP}$-hard. Similarly to Checkuri et al.~\cite{ChekuriSOS07}, we give a reduction from the minimum expansion problem. Interestingly, the hardness result applies even if the input graph $G$ is a tree. 

Letting $\mathcal{C}_K = \{S : S \neq \emptyset, S \subset K, \; |S| \leq k / 2 \}$ denote the set of terminal cuts, we can restate definition of cut sparsifier in the case $V_H = K$ as follows: for a graph $G=(V,E,c)$ with $K \subset V$, and $H=(K,E_H,c_H)$, we say that $H$ is a \emph{cut sparsifier} of $G$ if $\forall S \in \mathcal{C}_K,~ \sum_{e \in \delta(S)} c_H (e) \geq \mincut_G(S, K \setminus S)$. Now we define the decision variant of the minimum expansion and the cut sparsifier problem.

\medskip 

\noindent
\textbf{The Minimum Expansion Problem. }Given a graph $H=(K,E_H,c_H)$ and some positive constant $\alpha$, decide if there exists a subset $S \in \mathcal{C}_K$ such that $\sum_{e \in \delta(S)} c_H (e) / |S| < \alpha$.

\medskip 

\noindent
\textbf{The Cut Sparsifier Problem. } Given a graph $G=(V,E,c)$ with $K \subset V$, and $H=(K,E_H,c_H)$, decide if $H$ a cut sparsifier of $G$.

For convenience, we reformulate the cut sparsifier problem using the notion of cut polytopes, which we define below.

\begin{definition} For a given graph $G=(V,E,c)$ with $K \subset V$, we define $P_{\text{cut}}(G)$ to be the polytope containing all cut sparsifiers of $G$, i.e.,
\[
	P_{\text{cut}}(G) := \{u \in \mathbb{R}_{\geq 0}^{\binom{k}{2}} : \forall S \in \mathcal{C}_{K}, ~ \sum_{e \in \delta(S)} u (e) / \mincut_G(S,K \setminus S) \geq 1\} \enspace.
\]
\end{definition}

Before proceeding, we observe a simple fact. Let $G=(K \cup \{v\},E)$ be a star, where each edge has capacity $\alpha > 0$. Then the symmetric structure of $G$ gives the following:

\begin{fact} \label{star: fact}
Let $S \in \mathcal{C}_K$ be any terminal cut and $G$ be the star graph defined above. Then, the minimum terminal cut of $S$ in $G$ equals the \emph{scaled} cardinality of set $S$, i.e., $h_G(S) = \alpha \cdot |S|$. 
\end{fact}

\begin{theorem} Given a star $G = (K \cup \{v\},E)$ with uniform edge capacities $\alpha$, and some graph $H=(V,E_H,c_H)$, deciding whether $c_H \in P_\text{cut}(G)$ is co-$\mathcal{NP}$-hard.  
\end{theorem}
\begin{proof}
Given an instance of the Minimum Expansion Problem, i.e., a graph $H=(K,E_H,c_H)$ and some positive constant $\alpha$, we construct an instance of the Cut Sparsifier Problem by building a star graph $G=(K \cup \{v\}, E)$, where each edge has capacity $\alpha$ and letting $H$ be the candidate cut sparsifier of $G$. We claim that $H$ has expansion strictly less than $\alpha$ iff $c_H \not\in P_\text{cut}(G)$. 

Indeed, by Fact \ref{star: fact}, if $G$ has expansion strictly less than $\alpha$, then there exists a set $S \in \mathcal{C}_K$ such that $\sum_{e \in \delta(S)} u(e) / |S| < \alpha$. This implies that $\sum_{e \in \delta(S)} u(e) / \alpha |S| < 1$, and thus $c_H \not\in P_\text{cut}(G)$. The other direction is symmetric.

Since it is known that the decision variant of the Minimum Expansion Problem is $\mathcal{NP}$-hard (see~\cite{LeightonR99}), the co-$\mathcal{NP}$-hardness of the Cut Sparsifier Problem follows.
\end{proof}
\subsection{Hardness of Single-Source Flow Sparsifiers}

Similarly to the previous section, we can define \emph{the Flow Sparsifier Problem}: Given a graph $G=(V,E,c)$ with $K \subset V$, and $H = (K,E_H,c_H)$ decide if $H$ is a flow sparsifier of $G$, or equivalently,  if
\begin{equation} \label{eq: flow}
\forall d \in \mathbb{R}_{+}^{\binom{k}{2}},~ \con_H(d) \leq \con_G(d) \enspace.
\end{equation} 
In the following we give an equivalent definition for the above problem, which makes clearer the connection between our problem and a variant of the robust network design problem~~\cite{ChekuriSOS07}. For a graph $G$ defined as above, we let \[P_G=\{d \in \mathbb{R}_{+}^{\binom{k}{2}} : \con_G(d) \leq 1\} \enspace,\] be the demand polytope of $G$, namely, the polytope consisting of all demands that can be feasibly routed in $G$. We associate with $P_G$ a convex region $\mathcal{U}(P_G)$. For a given graph $H=(K,E_H,c_H)$, we say that $c_H$ is in $\mathcal{U}(P_G)$, if for each demand $d \in P_G$, there exist a flow that can \emph{feasibly route} $d$ in $H$. Now, it is clear that instead of asking  (\ref{eq: flow}), we can alternatively ask whether $c_H \in \mathcal{U}(P_G)$. The latter is exactly the separation problem of \emph{The Robust Network Design Problem} (abbr. RND). 

Checkuri et al.~\cite{ChekuriSOS07} have shown that the separation problem of some special version of RND is co-$\mathcal{NP}$-hard. Below we perform some modifications to adapt our problem to theirs. 

First, we remark that in the original definition of flow sparsifiers due to Leighton and Moitra~\cite{leighton}, the entries of some demand vector in some input graph $G$ can be positive for every undirected pair of terminals. Here, we will assume that we are given a distinguished terminal $r$ and the only positive entries of a demand vector $d$ are those that correspond to pairs involving $r$, i.e., $d_{r,x} > 0, $ for all $x \in K \setminus \{r\}$. This variant naturally leads to the notion of \emph{single-source} flow sparsifiers. 

Next, given a $k$-dimensional vector $b$, we need to construct an instance of single source flow sparsifier for a graph $G$ with demand polytope $P_G$ defined as follows (see Section $2$ in \cite{ChekuriSOS07})
\begin{equation} \label{flow: polytope}
\begin{aligned}
	\sum_{x \neq r} d_{rx} & \leq b_r &  \\
	d_{rx} & \leq b_x & \forall x \neq r, \\
    d_{x'x} & = 0 &   x' \neq r, \\
           d & \geq 0 & \enspace.
\end{aligned}
\end{equation}
To this end, given $b$, define the tree $G = (K \cup \{v\},E)$ with $K \subset V$, where $r$ is the distinguished terminal and $E= \{(x,v) : x \in K \setminus \{r\}\} \cup \{(x,v)\}$. We assign capacity $b_x$ to the edge $(x,v)$, for all $x \in K \setminus \{r\}$, and capacity $b_r$ to the edge $(r,v)$. Now, since $G$ is a tree and the routing paths are unique, one can easily observe that the demand polytope $P_G$ in the single-source flow sparsifier problem is exactly the polytope given in (\ref{flow: polytope}). Thus, it follows that the Single-Source Flow Sparsifier Problem is equivalent to the separation problem of the Single-Source RND. 

Checkuri et al.~\cite{ChekuriSOS07} devised a hardness result for the Single-Source RND, which by the above equivalence leads to the same hardness result for the Single-Source Flow Sparsifier Problem:

\begin{theorem}
Given a tree $G=(V,E,c)$ with $K \subset V$ defined as above and some graph $H = (K,E_H,c_H)$, deciding whether $c_H \in \mathcal{U}(P_G)$ is co-$\mathcal{NP}$-hard.
\end{theorem}
 	
\section{Improved Results for Quasi-Bipartite Graphs} \label{section: applications}
In this section, we present two new tradeoffs for flow sparsifiers in
quasi-bipartite graphs. For this family of graphs, Andoni et al.~\cite{andoni}
show how to obtain flow sparsifier with very good quality and moderate size.
Specifically, they obtain an $(1+\varepsilon)$-quality flow sparsifier of size
$\widetilde{O}(k^{7}/\varepsilon^3)$. In the original definition of flow
sparsifiers, Leighton and Moitra~\cite{leighton} studied the version where
sparsifiers lie only on the terminals, i.e., $V_H = K$. For this restricted
setting, we obtain a flow sparsifier of quality $2$.

Exact Cut Sparsifier (a.k.a Mimicking Networks) were introduced by Hagerup et
al.~\cite{HagerupKNR98}. In their work they show that general graphs admit
exact cut sparsifiers of size doubly exponential in $k$. As a second result, we
show that unit weighted quasi-bipartite graphs admit an exact flow sparsifier
of size $2^{k}$.

A graph $G$ with terminals $K$ is \emph{quasi-bipartite} if the non-terminals
form an independent set. Throughout this section we assume w.l.o.g.\ that we are
given a bipartite graph with terminals lying on one side and non-terminals in
the other (this can achieved by subdividing terminal-terminal edges).

\paragraph{\mathversion{bold}A $2$-quality flow sparsifier of size $k$.\mathversion{normal}}
Assume we are given an unweighted bipartite graph $G$ with terminals $K$. The
crucial observation is that we can view $G$ as taking union over stars, where
each non-terminal is the center connected to some subset of terminals. Lemma
\ref{lemma: Merging} allows us to study these stars independently. Then, for
every such star, we apply Lemma \ref{lemma: TreeTerminalLeaves} to obtain a
flow sparsifier only on the terminals belonging to that star. Finally, we merge
the resulting sparsifiers and construct a sparsifier $H$ with $V(H) = K$ by
another application of Lemma \ref{lemma: TreeTerminalLeaves}. Since the quality
of every star in isolation is $2$ or better, $H$ is also a $2$-quality flow
sparsifier.

We note that Lemma \ref{lemma: TreeTerminalLeaves} only works for unweighted
trees. There is an easy extension that gives a similar lemma for weighted
stars.
\begin{lemma} \label{lemma: WeightedStar}
Let $G=(K \cup \{u\}, E, c)$ be a weighted star with $k$ terminals. Then $G$ admits a $2$-quality flow sparsifier $H$ of size $k$.
\end{lemma} 
\begin{proof}
Let $C = \sum_{x=1}^{k} c(u,i)$ be the sum over all edge capacities in $G$. Note that by contracting the star edge $(u,x)$ we get a flow sparsifier $H_{x}$ of quality $\sum_{x' \neq x} c(u,x') / c(u,x)$. There are at most $k$ such sparsifiers. Now we construct a sparsifier $H$ where the edge $(u,x)$ is contracted with probability $c(u,x)/C$. Equivalently,  $H = \sum_{x=1}^{k} (c(u,x)/C) \cdot H_{x}$ by Lemma \ref{lemma: ConvexComb}.

We observe that $H$ is a complete graph on the terminals, where $c_{H}(x,x') = 2 \cdot  c(u,x) \cdot c(u,x') / C$. By Lemma \ref{lemma: Moitra}, routing the demand $d_{H}$ in $G$ gives the following upper bound on the congestion of any edge $(u,x)$ in $G$:
\begin{equation*}
	2 \left(\frac{ c(u,x) \cdot \sum_{x' \neq x} c(u,x') }{c(u,x) \cdot C}\right)  = 2 \left(\frac{ c(u,x) \cdot (C - c(u,x)) }{c(u,x) \cdot C}\right)
	 = 2 \left(1 - \frac{c(u,x)}{C}\right) \enspace.
\end{equation*}
The latter implies that $H$ is a $2$-quality flow sparsifier for $G$.
\end{proof}

Applying the decomposition and merging lemma similarly to the unweighted case leads to the following theorem:
\begin{theorem} Let $G = (V,E,c)$ with $K \subset V$ be a weighted
  quasi-bipartite graph. Then $G$ admits a $2$-quality flow sparsifier $H$ of
  size $k$.
\end{theorem}

\paragraph{An exact flow sparsifier of size \mathversion{bold}$2^{k}$\mathversion{normal}.}
In what follows it will be convenient to work with an equivalent definition for
Flow Sparsifiers. Let $\lambda_G(d)$ denote the maximum fraction of
\emph{concurrent flow} when routing demand $d$ among terminals in graph $G$.
Then $H=(V_H,E_H,c_H)$ with $K \subset V_H$ is a \emph{flow sparsifier} of $G$
with quality $q \geq 1$ if for all demand functions $d$, $\lambda_G(d) \leq
\lambda_H(d) \leq q \cdot \lambda_G(d)$.

The high level idea of our approach is to create ``types'' for non-terminals
and then merge all non-terminals of the same type into a single non-terminal
(i.e., add infinity capacity among all non-terminals of the same type). The
main difficulty is to define the right types and show that the merging does not
affect the multi-commodity flow structure among the terminals. A similar
approach was developed by Andoni et al.~\cite{andoni}, but their guarantees
applies only to approximate flow sparsifier.

We start by defining types. We say that two non-terminals $u,v$ are of the same
type if they are incident to the same subset of terminals. Non-terminals of the
same type form groups. Note that a non-terminal belongs to an unique group. The
\emph{size} of the group is the number of non-terminals belonging to that
group. Since the set of non-terminals is an independent set, by Lemma
\ref{lemma: Merging}, we can construct sparsifiers for each group
independently. Our final sparsifier is obtained by merging the sparsifiers over
all groups. By another application of Lemma \ref{lemma: Merging}, if the
sparsifiers of the groups are exact flow sparsifiers, then the final sparsifier
is also an exact flow sparsifier for the original graph.

Next, if we replace each group by a single non-terminal, then the size
guarantee of the final sparsifier follows from the fact that there are at most
$2^{k}$ different subsets of terminals. Below we formalize the merging
operation within groups.

Let $G_i = (K' \cup \{v_1, \ldots, v_{n_i}\}, E_i, c)$ be a group of size $n_i
\geq 2$, where $E_i = \{\{v_j,x\}: j \in \{1,\ldots,n_i\}, \; x \in K'\}$, $K'
\subseteq K$ and $c(e) = 1$, $e \in E_i$. We get:
\begin{lemma} \label{lemma: Group} Let $G_i$ with $K' \subset V(G_i)$ be a
  group of size $n_i \geq 2$ defined as above. Then $G_i$ can be replaced by a
  star $H_i = (K' \cup \{v_1\}, E_{H_i}, c_{H_i})$ with edge weights
  $c_{H_i}(e) = n_i$, for all $e \in E_{H_i}$, and which preserves exactly all
  multicommodity flows between terminals from $K'$.
\end{lemma}
Taking the union over all sparsifiers $H_i$ 
leads to the following theorem:
\begin{theorem}
  Let $G = (V,E)$ with $K \subset V$ be a unit weighted quasi-bipartite graph.
  Then $G$ admits an exact flow sparsifier $H$ of size at most $2^{k}$.
\end{theorem}
\begin{proof}[Lemma \ref{lemma: Group}] 
  First, observe that we can think of $H_i$ as adding infinity capacity edges
  between non-terminals in $G_i$. Then merging into a single non-terminal is
  done by simply adding edge weights incident to the same terminal. More
  precisely, let $E_{H_i} = \{(v_r, v_s) : r,s = 1,\ldots, n_i, \; r \neq s\}$.
  Then, we can assume that $H_i = (K' \cup \{v_1, \ldots v_{n_i}\}, E_i \cup
  E_{H_i}, c_{H_i})$ where $c_{H_i} (e) = c(e)$ if $e \in E_i$ and $c_{H_i}(e)
  = \infty$ if $e \in E_{H_i}$.

  Since we can route every feasible demand from $G_i$ in $H_i$ even without
  using the infinity-capacity edges, it is immediate that for any demand function
  $d$, $\lambda_{H_i}(d) \geq \lambda_{G_i}(d)$. Thus, we only need to show
  that $\lambda_{H_i}(d) \leq \lambda_{G_i}(d)$. To achieve this, we will use
  the dual to the maximum concurrent flow problem (i.e., the Fractional
  Sparsest Cut Problem). The dual problem is the following\footnote{Note that
    the dual requires that $\delta_{st}$ is at most the length of the shortest
    $s$-$t$ path. In our scenario this is always a $2$-hop path. Hence, the
    above formulation is correct.}:
\begin{equation} \label{2:eq9}
\begin{aligned}
& \text{min} & & \sum_{j=1}^{n_i} \nolimits \sum_{x \in K'} \nolimits \ell_{v_jx} \\
& \text{s. t.} & &\ell_{sv_j} + \ell_{v_jt}  \geq \delta_{st}  && \forall \{s,t\} \in \tbinom{K'}{2}, \; \forall j \in \{1,\ldots,n_i\} \\
& &&\sum_{ \{s,t\} \in \binom{K'}{2}} \nolimits d_{st} \delta_{st} \geq 1 \\
& & & \ell_e \geq 0, \quad \delta_{st} \geq 0 \enspace.
\end{aligned}
\end{equation}
Let $d$ be an arbitrary demand function. Moreover, let $\{\ell_e, \delta_{st}\}$
be an optimal solution of value $\lambda_{G_i}(d)$ for the LP in Eqn.
(\ref{2:eq9}), where $\delta_{st}$ is the shortest-path distance induced by the
length assignment $\ell$. We first modify this solution and get a new feasible
solution with the same cost and a certain structure that we will later exploit.

The modification works as follows. For every terminal we create a set of edges incident to that terminal. Then, within each set, we
replace the length of each edge by the total \textit{average} length of the
group. Specifically, for every $x \in K'$, let $E_{x} = \{(v_j, x) : j = 1,
\ldots, n_i\}$ be the set of edges incident to $x$. 

The new edge lengths are defined as follows: $\widetilde{\ell}_{v_jx} = {\sum_{e \in E_{x}} \nolimits \ell_{e}}/{n_i}, \forall x \in K', \forall j = 1,\ldots,n_i$.
Let $\smash{\widetilde{\delta}_{st}}$ be the new shortest-path distance induced by the length assignment $\widetilde{\ell}$. In order for $\{\widetilde{\ell}_e, \widetilde{\delta}_{st}\}$ to be feasible, we need to show that $\widetilde{\delta}$ dominates $\delta$, i.e., $\widetilde{\delta}_{st} \geq \delta_{st}$, for every pair $s,t \in K'$. Indeed, since edge lengths within groups are the same, we get that for every pair $s,t \in K'$:
\begin{equation*}
\begin{aligned}
	\widetilde{\delta}_{st} = \widetilde{\ell}_{sv_1} + \widetilde{\ell}_{v_1t} & = \frac{1}{n_i}{\sum_{e \in E_{s}} \nolimits \ell_{e} + \displaystyle\frac{1}{n_i}\sum_{e \in E_{t}} \nolimits \ell_{e}} =  \frac{1}{n_i}{\sum_{j = 1}^{n_i} \nolimits \left(\ell_{sv_j} + \ell_{v_jt}\right)}\\
	& \geq \min_{j \in \{1,\ldots,n_i\}}\{\ell_{sv_j} + \ell_{v_jt}\} \geq \delta_{st} \enspace.
\end{aligned}
\end{equation*}
Additionally, observe that the new solution has the same optimal value, namely 
\[\lambda^{*}_{G'_i}(d) = \sum_{j=1}^{n_i}{\sum_{x \in K'}} \ell_{v_jx} = \sum_{j=1}^{n_i}{\sum_{x \in K'}} \widetilde{\ell}_{v_jx}\enspace.\]
Hence, we can assume without loss of generality that an optimal solution satsifies:
$\smash{\widetilde{\ell}_{v_1x}} = \ldots = \smash{\widetilde{\ell}_{v_{n_i}}x}, \; \forall x \in K'$.
Now, we add edges $(v_i,v_j)$ to $G_i$ and set $\widetilde{\ell}_{v_i v_j} =
0$, for all $i,j = 1,\ldots, n_i$. Note that shortest-path distances
$\widetilde{\delta}_{st}$ do not change by this modification. Therefore, by
adding these zero edge lengths between the non-terminals, we still get an
optimum solution $\{\widetilde{\ell}_e, \widetilde{\delta}_{st}\}$ for the LP
in (\ref{2:eq9}).

Finally, let us define the dual problem for the star $H_i$:
\begin{equation} \label{2:eq12}
\begin{aligned}
& \text{min} & & \sum_{j=1}^{n_i} \nolimits \sum_{x \in K'} \nolimits \ell_{v_jx} \\
& \text{s. t.} & &  \sum_{e \in P_{st}} \nolimits \ell_e   \geq \delta_{st}  
\quad\quad\forall \{s,t\} \in \tbinom{K'}{2}, \;\forall s\text{-}t \text{ paths on } E \cup E_{H_i} \\
& & &\sum_{ \{s,t\} \in \binom{K'}{2}} \nolimits d_{st} \delta_{st} \geq 1 \\
& & & \ell_e \geq 0, \quad \delta_{st} \geq 0,\quad  \forall e \in
E_{H_i}\;\ell_e = 0\enspace.
\end{aligned}
\end{equation}
It follows from above that $\{\widetilde{\ell}_e, \widetilde{\delta}_{st}\}$ is a feasible solution for the LP in (\ref{2:eq12}). Hence, $\lambda_{H_i}(d) \leq \lambda_{G_i}(d)$, what we were after.
\end{proof}


\bibliographystyle{plainurl}
\bibliography{literature}

\appendix

\section{Missing Proofs}
We first state the following simple fact:
\begin{fact} \label{2:fact1}
If $a_1,\ldots,a_k$ and $b_1,\ldots,b_k$ are positive numbers, then
\[
	\frac{\sum_{i}^{k} a_i}{\sum_{i}^{k} b_i} \leq \max_{i = 1,\ldots, k} \frac{a_i}{b_i}.
\] 
\end{fact}

\begin{pfof}{Lemma \ref{lemma: ConvexComb}}
Let $d$ be an arbitrary demand function. Since by assumption $H_i$'s are vertex flow sparsifiers, by definition it follows that $\con_{H_i}(d) \leq \con_{G}(d),$ $i = 1,\ldots, m$. Fix a flow $f^{i}$ and its corresponding decomposition $D^{i} = \{(p_1, f^{i}_{p_1}), (p_2, f^{i}_{p_2}) ,\ldots\}$ for routing $d$ in $H_i$, $i = 1,\ldots,m$, where $p_{\ell}$ is a path with terminal endpoints and $f^{i}_{p_{\ell}}$ is the amount of flow sent along this path. We now scale each flow decomposition $D^{i}$ and the capacities of $H_i$'s by the multiplier $\alpha_i$, $i =1, \ldots, m$. Finally, take the union over all scaled sparsifiers, i.e., $H' = \sum_{i} \alpha_i H_i$, along with their decompositions. This can be seen as re-routing the demand $d$ since $d_{st} = \sum_{i} \alpha_i d_{st} = \sum_i \sum _{P_{st} \in D^{i}} \sum_{p \in P_{st}} \alpha_i f^{i}_{p}$, for all terminal pairs $s,t \in K$. 

We need to show that $\con_{H'}(d) \leq \con_{G}(d)$. Indeed, fix an arbitrary edge $e'$ from $H'$. For $i=1,\ldots,m$, let $f^{i}(e') = \sum_{p \in D^{i}: \; e' \in p} f_p$ denote the total flow sent along edge $e'$. The congestion of $e'$ is:
\[
	\frac{\sum_{i : e' \in E_i} \alpha_i f^{i}(e')}{\sum_{i : e' \in E_i} \alpha_i c_i(e')} \leq \max_{i : e' \in E_i} \frac{\alpha_i f^{i}(e')}{\alpha_i c_{i}(e')} \leq \con_{G}(d),
\]
where the first inequality follows from the Fact \ref{2:fact1} and the last one from $H_i$ being a flow sparsifier for $G$. Since $e'$ was chosen arbitrarily, it follows that $\con_{H'}(d) \leq \con_{G}(d)$. Moreover, the fact that $d$ was chosen arbitrarily implies that $H'$ is a vertex flow sparsifier for $G$. 
\end{pfof}

\end{document}